\documentclass{jpc} 


\usepackage{hyperref}
\usepackage{natbib}
\usepackage[ruled]{algorithm2e}
\usepackage{verbatim}   
\usepackage{bm}   

\theoremstyle{plain} 

\begin{document}

\title[Bounded Gaussian Mechanism]{The Bounded Gaussian Mechanism \\ for Differential Privacy}

\author[Bo Chen]{Bo Chen}	
\address{University of Florida, Gainesville, Florida, USA}	
\email{bo.chen@ufl.edu}  

\author[Matthew Hale]{Matthew Hale}	
\address{University of Florida, Gainesville, Florida, USA}	
\email{matthewhale@ufl.edu}  
\thanks{This work was supported by the NSF under CAREER Grant 1943275,
by AFOSR under grant FA9550-19-1-0169, and by ONR under
grant N00014-21-1-2502.}	





\begin{abstract}
  \noindent  The Gaussian mechanism is one differential privacy mechanism
  commonly used to protect numerical data. However, it may be ill-suited to
  some applications because it has unbounded support and thus can produce
  invalid numerical answers to queries, such as negative ages or
  human heights in the tens of meters. One can project such private values
  onto valid ranges of data, though such projections lead to the accumulation of
  private query responses at the boundaries of such ranges, thereby harming
  accuracy. Motivated by the need for both privacy and accuracy over bounded
  domains, we present a bounded Gaussian mechanism for differential privacy, which
  has support only on a given region. We present both univariate and multivariate
  versions of this mechanism and illustrate a significant reduction in variance
  relative to comparable existing work. 
\end{abstract}

\maketitle

\section*{Introduction}\label{S:one}
As many engineering applications have become increasingly reliant on user data, data privacy has become a concern that data aggregators and curators must take into consideration. In numerous applications, such as
healthcare~\cite{YANG20181437}, energy systems~\cite{Asghar2017}, transportation systems~\cite{Zhang2018} and Internet of Things (IoT)~\cite{Medaglia2010},
the data gathered to support system operation
often contains sensitive individual information. Differential privacy~\cite{Dwork2014} has emerged as a standard privacy framework
that can be used in such applications to protect sensitive data while allowing
privatized data to remain useful. 
Differential privacy is a statistical notion of privacy that provides privacy guarantees 
by adding carefully calibrated noise to sensitive data or functions of sensitive data. 
Its key features include: (i) it is robust to side information, in that any additional knowledge about data-producing entities does not weaken 
their privacy by much~\cite{Kasiviswanathan_Smith_2014}, and (ii) it is immune to post-processing, in that any transformation of private data stays private~\cite{Dwork2014}. 

Well-known mechanisms for the enforcement of differential privacy include 
the Laplace mechanism~\cite{Dwork2006}, the Gaussian mechanism~\cite{Dwork2014}, 
and the exponential mechanism~\cite{McSherry2007,McSherry2009}. Other mechanisms have 
been developed for specific applications, in some cases
by building upon or modifying these
well-known mechanisms. A representative sample
includes the Dirichlet mechanism on the unit
simplex~\cite{Gohari2021}, mechanisms for sensitive words
of symbolic data~\cite{chen2022}, the matrix-variate Gaussian
mechanism for matrix-valued queries~\cite{chanyaswad2018},
and the XOR mechanism for binary-valued data~\cite{ji2021}. 
As the use of differential privacy has grown, these and other mechanisms
have been needed to respond to privacy needs in settings with
constraints on allowable data. 


The Gaussian mechanism has been used for some classes of numerical data, although it may also require modifications for some types
of sensitive data because the Gaussian mechanism adds unbounded noise. For example, 
the Gaussian mechanism may generate 
negative values for data such as ages, salaries, and weights, and
such negative values are not meaningful in these contexts. 
One attempt to solve this problem is through 
projecting out-of-domain results back to the closest value in the given domain. 
Although this procedure does not weaken differential privacy because
it is only post-processing on private data, 
it has been observed to lead to low accuracy in applications~\cite{holohan2018bounded} and
thus is undesirable. Nonetheless, the Gaussian mechanism has been
used in numerous applications, including deep learning~\cite{Abadi2016}, \cite{Yu2019}, convex optimization~\cite{Bassily2014}, filtering and estimation problems~\cite{Le2014} and cloud control~\cite{Hale2018},
all of which can use data that is inherently bounded in some way. 

On the other hand, compared to the Laplace mechanism, which is another popular privacy mechanism designed for 
numerical data, the Gaussian mechanism is able to use lower-variance noise to provide the same privacy level for 
high-dimensional data. This is because the variance of privacy noise 
is an increasing function of the sensitivity of 
a given query, and the Gaussian mechanism allows the use of the~$L_2$ sensitivity which, in 
applications like 
federated learning~\cite{Wei2020} and deep learning~\cite{Abadi2016}, is much lower than the $L_1$ sensitivity used 
in Laplace mechanism.

Therefore, in this paper we develop a new type of Gaussian mechanism that 
accounts for the boundedness of semantically valid data in a given application. 
In the univariate case we consider data confined to a closed interval, and
in the multivariate case we consider data confined to a product of closed
intervals. For both cases we show that the bounded Gaussian
mechanism provides~$\epsilon$-differential privacy, rather
than~$(\epsilon, \delta)$-differential privacy as is provided by the
ordinary Gaussian mechanism. We also present algorithms for finding
the minimum variance necessary to enforce~$\epsilon$-differential privacy
using the bounded Gaussian mechanism. The mechanisms we develop do not rely on projections
and thus avoid the harms to accuracy that can accompany projection-based approaches. 

Besides, \cite{Balle2018} point out that the original Gaussian mechanism's variance bound is far from tight when the privacy parameter $\epsilon$ approaches either $0$ or $\infty$. This is because the original Gaussian mechanism calibrates noise with a worst-case bound which is only tight for a small range of $\epsilon$. In this work, we present 
bounded mechanisms that address this limitation by leveraging the boundedness of the domains we consider. 
This boundedness excludes the tail of each Gaussian distribution from our analysis and allows for the development of tight bounds
on the variance of noise required for all~$\epsilon$. 

Related work in~\cite{holohan2018bounded} developed a bounded Laplace mechanism,
and in this work we bring the ability to bound private data to the
Gaussian mechanism.  Developments in~\cite{Liu2019} include
a generalized Gaussian mechanism with bounded support, and we show
in Section~\ref{S:algo} that the mechanism developed in this paper
requires significantly lower variance of noise to attain
differential privacy and thus provides improved accuracy. 

\textbf{Notation} Let $\mathbb{N}$ denote the set of all positive integers and 
let $I$ denote the identity matrix.
We use non-bold letters to denote scalars, e.g., $a\in\mathbb{R}$, and we use
bold letters denote vectors, e.g., $\mathbf{b}\in\mathbb{R}^n$. 
All intervals of the form~$[a, b]$ are assumed to have~$a < b$. For~$\mathbf{a},\mathbf{b} \in \mathbb{R}^m$
we define~$[\mathbf{a}, \mathbf{b}] = [a_1, b_1] \times [a_2, b_2] \times \cdots \times [a_m, b_m]$.
We also use the function
\begin{equation*}
    \text{erf}(z) = \frac{2}{\sqrt{\pi}}\int_0^z e^{-t^2}dt.
\end{equation*}

\section{Preliminaries and Problem Statement}

This section gives background on differential privacy and problem statements that are the focus of the remainder of the paper.

\subsection{Differential Privacy Background}

Differential privacy is enforced by a \emph{mechanism}, which is a randomized map. 
We let~$S$ denote the space of sensitive data of interest. 
For nearby pieces of sensitive data, a mechanism must produce outputs that are approximately distinguishable. The definition of “nearby” is given by an adjacency relation.

\begin{defi}[Adjacency]
    Fix $k\in\mathbb{N}$. Two databases $d,d'\in S^n$ are adjacent if they differ in $k$ rows.
\end{defi} 

For numerical queries $Q:S^n\rightarrow \mathbb{R}^m$, $m \in \mathbb{N}$, the sensitivity of a query~$Q$, denoted~$\Delta Q$, is used to calibrate
the variance of privacy noise used to implement privacy. 
Throughout this paper we will use the~$\ell_2$-sensitivity.
\begin{defi}[$\ell_2$-Sensitivity]\label{def:sensitivity}
    The $\ell_2$-sensitivity of a query $Q:S^n\rightarrow\mathbb{R}^m$ is
    \begin{equation*}
        \Delta Q := \max_{\text{adjacent }d,d'}||Q(d)-Q(d')||_2.
    \end{equation*}
\end{defi}

The $\ell_2$ sensitivity captures the largest
magnitude by which the outputs of~$Q$ can change
across two adjacent input databases. We next introduce the definition of differential privacy itself.

\begin{defi}[Differential Privacy]
    Fix a probability space~$(\Omega, \mathfrak{F}, \mathbb{P})$, 
    an adjacency parameter~$k \in \mathbb{N}$, a privacy parameter $\epsilon>0$, and a query $Q:S^n\rightarrow\mathbb{R}^m$. A mechanism $M: \Omega \times S^n \to \mathbb{R}^m$ is $(\epsilon,\delta)$-differentially private if for all measurable sets $A\subseteq\mathbb{R}^m$ and for all adjacent databases $d,d'\in S^n$ it satisfies 
    $\mathbb{P}\big[M(Q(d))\in A\big] \leq e^\epsilon\mathbb{P}\big[M(Q(d'))\in A\big]+\delta$.
\end{defi}

The privacy parameter $\epsilon$ sets the strength of privacy protections, and a smaller $\epsilon$ implies stronger privacy. 
The parameter $\delta$ can be interpreted as a relaxation parameter, in the sense that
it is the probability that~$\epsilon$-differential privacy fails to hold~\cite{Beimel2013}
(though it may hold with a different value of~$\epsilon$). 
If $\delta=0$, then a mechanism is said to be $\epsilon$-differentially private.

In the literature, $\epsilon$ has ranged from $0.01$ to $10$~(\cite{Hsu_2014}). A differential privacy mechanism guarantees that the randomized outputs of 
two adjacent databases will be made approximately indistinguishable to any recipient of their privatized forms, including any eavesdroppers. 
One widely used differential privacy mechanism is the Gaussian mechanism. The standard Gaussian mechanism is given by the following theorem:

\begin{thm}[Standard Gaussian Mechanism]
    Let $\epsilon\in(0,1)$ and~$\delta \in (0, 1]$ be given. Fix an adjacency parameter~$k \in \mathbb{N}$. 
    For $c^2>2\ln{(1.25/\delta)}$ and a query $Q$ of the database $d$, the Gaussian mechanism $M_G(Q(d))=Q(d)+\eta$, $\eta \sim \mathcal{N}(\mathbf{0},\sigma^2I)$ with parameter $\sigma^2\geq c\Delta Q/\epsilon$ is $(\epsilon,\delta)$-differentially private.
\end{thm}

As discussed in the introduction, the Gaussian mechanism has been favored
over the Laplace mechanism in several applications, including deep learning,
empirical risk minimization, and various problems in control
theory and optimization. Simultaneously, these applications may have bounds
on the data they use, such as bounds on training data for a learning algorithm or
bounds on states in a control system, though the Gaussian mechanism has infinite
support and hence produces unbounded private outputs.

%

\subsection{Problem Statement}
In this work we seek a Gaussian
mechanism that respects given bounds on data, and we formalize the development
of this mechanism in the next two problem statements. 

\begin{prob}[Univariate bounded Gaussian mechanism]\label{prob:univ_bd_gaus}
    Given a query $Q:S^n\rightarrow D$, where $D=[a,b]\subset\mathbb{R}$ ($a<b$, both finite)
    is a constrained domain, and a privacy parameter~$\epsilon > 0$, 
    develop a mechanism $M_B: \Omega \times S^n \to D$
    that is an $\epsilon$-differentially private approximation of~$Q$ 
    and generates outputs in~$D$ with probability~$1$. 
\end{prob}

\begin{prob}[Multivariate bounded Gaussian mechanism]\label{prob:multi_bd_gaus}
    Given a query $Q:S^n\rightarrow D$, where $D=[\mathbf{a},\mathbf{b}]\subset\mathbb{R}^m$ 
    ($a_i<b_i$, both finite for all $i$)
    is a constrained domain, 
    and a privacy parameter~$\epsilon > 0$, 
    develop a mechanism $M_B: \Omega \times S^n \to D$    
    that is an $\epsilon$-differentially private approximation of~$Q$ 
    and generates outputs in~$D$ with probability~$1$. 
\end{prob}

We solve Problem~\ref{prob:univ_bd_gaus} in Section~\ref{sec:bounded_gaussian_mechanism},
and we solve Problem~\ref{prob:multi_bd_gaus} 
in Section~\ref{sec:multivariate_mechanism}. 



\section{Univariate Bounded Gaussian mechanism}\label{sec:bounded_gaussian_mechanism}
In this section, we develop the bounded Gaussian mechanism for a bounded domain $D=[a,b]\subset\mathbb{R}$. 
We now give a formal statement of the univariate bounded Gaussian mechanism, and this
will be the focus of the rest of this section: 

\begin{defi}[Univariate Bounded Gaussian Mechanism]\label{def:bounded_guassian}
    Fix~$\sigma > 0$. 
    Let a query~$Q : S^n \to D$ be given where $D=[a,b]\subset\mathbb{R}$, and suppose that~$d \in S^n$
    generates the output~$Q(d) = s \in D$. 
    Then the univariate bounded Gaussian mechanism $M_B:\Omega\times S^n\to D$
    is given by the probability density function
\begin{equation}
    p_B(x)=\begin{cases}
        \frac{1}{\sigma}\frac{\phi\left(\frac{x-s}{\sigma}\right)}{\Phi\left(\frac{b-s}{\sigma}\right)-\Phi\left(\frac{a-s}{\sigma}\right)} & \text{if } x\in D,\\
        0 & \text{otherwise,}
    \end{cases}
\end{equation}

where
\begin{align*}
    &\phi(x) = \frac{1}{\sqrt{2\pi}}\exp\left(-\frac{1}{2}x^2\right), \\
    &\Phi(x) = \frac{1}{2}\left(1+\text{erf}(x/\sqrt{2})\right). 
\end{align*}

\end{defi}

Since $Q(d)=s$, the density function $p_B(x)$ is not the same for each database, and data-dependent noise 
is known to be problematic in general for differential privacy~\cite{Nissim2007}.
Therefore using the parameters of the standard Gaussian mechanism 
is no longer guaranteed to satisfy differential privacy for the bounded Gaussian mechanism. We next provide the parameters 
required for the bounded Gaussian mechanism to provide differential privacy by first introducing some preliminary results.

\subsection{Preliminary Results}
\begin{lem}\label{lem:preli_result_1}
    Let~$\sigma>0$ be given. For a database $d\in S^n$    
    and a query~$Q$ with $Q(d)=s\in D$ and $D=[a,b]\subset \mathbb{R}$, let
    \begin{align}
        C(s,\sigma) = \frac{1}{\Phi\left(\frac{b-s}{\sigma}\right)-\Phi\left(\frac{a-s}{\sigma}\right)}.\label{eq:definition_C}
    \end{align}
    Let $d'\in S^n$ be adjacent to $d$ in the sense of Definition~\ref{def:sensitivity}, and suppose that $Q(d')=s'\in D$. Then we have
    \begin{align*}
        \max_{s,s'} \frac{C(s,\sigma)}{C(s',\sigma)} = \frac{C(a,\sigma)}{C(a+c,\sigma)},
    \end{align*}
    where $c=|s'-s|\leq \Delta Q$.
\end{lem}
\begin{proof}
    See Appendix~\ref{apdx:proof_preli_result_1}.
\end{proof}

This leads to the following lemma.
\begin{lem}\label{lem:preli_result_2}
    Let $C(s,\sigma)$ be given in Equation~\eqref{eq:definition_C}. 
    Given two adjacent databases $d\in S^n$ and $d'\in S^n$ 
    and a query $Q:S^n\rightarrow D$, 
    suppose that 
    $Q(d)=s$ and $Q(d')=s'$ such that $s'=s+c$ and by the definition of adjacency we have $|c|\leq \Delta Q$.
    Then 
    \begin{equation*}
        \max_c \frac{C(a,\sigma)}{C(a+c,\sigma)}=:\Delta C(\sigma)=\begin{cases}
        \frac{C(a,\sigma)}{C(a+\Delta Q,\sigma)} & \text{if } \Delta Q\leq\frac{b-a}{2},\\[5pt]
        \frac{C(a,\sigma)}{C\left(\frac{b+a}{2},\sigma\right)} & \textnormal{otherwise.}
        \end{cases}
    \end{equation*}
\end{lem}
\begin{proof}
    See Appendix~\ref{apdx:proof_preli_result_2}.
\end{proof}

\subsection{Main Result on the Univariate Bounded Gaussian Mechanism}
We now proceed to the main result of this section, which bounds the variance 
required for the bounded Gaussian mechanism to provide~$\epsilon$-differential privacy. 
\begin{thm}\label{thm:main_result_1}
    Let~$\epsilon > 0$ and let a query~$Q : S^n \to D$ be given. 
    Suppose~$d\in S^n$ generates the output $Q(d)=s\in D$ and $D=[a,b]\subset \mathbb{R}$. 
    Then the Univariate Bounded Gaussian mechanism~$M_B$ 
    provides~$\epsilon$-differential privacy if
    \begin{equation}\label{eq:diff_pvt_one_d}
        \sigma^2\geq \frac{\left[(b-a)+\frac{\Delta Q}{2}\right]\Delta Q}{\epsilon-\ln(\Delta C(\sigma))}.
    \end{equation}
\end{thm}

\begin{proof}
    See Appendix~\ref{apdx:proof_main_result_1}.
\end{proof}

We observe that in Equation~\eqref{eq:diff_pvt_one_d} the privacy noise variance $\sigma$ appears on both sides of the equation. Given a value of~$\epsilon > 0$, it is desirable to use
the smallest~$\sigma$ as this corresponds to using the smallest variance of privacy
noise for a given privacy level. 
To do so, in the next section, we form a zero finding problem to find the 
minimum~$\sigma$ that satisfies Equation~\eqref{eq:diff_pvt_one_d}.

\subsection{Calculating \texorpdfstring{$\sigma$}{Lg}}\label{sec:calculate_sigma}
Let $\sigma^*$ be the smallest admissible value of $\sigma$. The calculation of $\sigma^*$ can be formulated 
as a zero finding problem, where we require
\begin{equation}
    f(\sigma^*)=(\sigma^*)^2 - \frac{\left[(b-a)+\frac{\Delta Q}{2}\right]\Delta Q}{\epsilon-\ln(\Delta C(\sigma^*))}=0.\label{eq:function_f}
\end{equation}

We now present some technical lemmas concerning the function $f$ with respect to a point
\begin{equation}
    \sigma_0=\sqrt{\frac{\left[(b-a)+\frac{\Delta Q}{2}\right]\Delta Q}{\epsilon}}.\label{eq:sigma_0}
\end{equation}

\begin{lem}\label{lem:denominator_greater_0}
    For any $\epsilon>0$, we have $\epsilon-\ln(\Delta C(\sigma_0))>0$.
\end{lem}
\begin{proof}
See Appendix~\ref{apdx:proof_denominator_greater_0}.
\end{proof}

\begin{lem}\label{lem:delta_c_greater_0}
    For all $\sigma\in(0,\infty)$, we have $\ln(\Delta C(\sigma))>0$.
\end{lem}
\begin{proof}
See Appendix~\ref{apdx:proof_delta_c_greater_0}.
\end{proof}

Next, Lemma~\ref{lem:lemma_f_sigma_0} shows the value of $f(\sigma_0)$ and Lemma~\ref{lem:lemma_f_decreasing} shows $f$ is a monotonically increasing function on $[\sigma_0,\infty)$.

\begin{lem}\label{lem:lemma_f_sigma_0}
    For the point $\sigma_0=\sqrt{\frac{\left[(b-a)+\frac{\Delta Q}{2}\right]\Delta Q}{\epsilon}}$, we have $f(\sigma_0)<0$.
\end{lem}
\begin{proof}
By plugging in $\sigma_0$ we have
\begin{equation*}
    f(\sigma_0)=\frac{\left[(b-a)+\frac{\Delta Q}{2}\right]\Delta Q}{\epsilon}-\frac{\left[(b-a)+\frac{\Delta Q}{2}\right]\Delta Q}{\epsilon-\ln(\Delta C(\sigma_0))}.
\end{equation*}

By Lemma~\ref{lem:denominator_greater_0} and Lemma~\ref{lem:delta_c_greater_0} we have $\epsilon>\epsilon-\ln(\Delta C(\sigma_0))>0$. Therefore $f(\sigma_0)<0$.
\end{proof}

\begin{lem}\label{lem:lemma_f_decreasing}
    For any $\sigma\in[\sigma_0,\infty)$, we have $f'(\sigma)>0$.
\end{lem}
\begin{proof}
See Appendix~\ref{apdx:proof_lemma_f_decreasing}.
\end{proof}

By combining Lemma~\ref{lem:lemma_f_sigma_0} and Lemma~\ref{lem:lemma_f_decreasing}, we can use Algorithm 1 in~\cite{holohan2018bounded}, given below as Algorithm~\ref{alg:zero_finding_algorithm},
to find the exact fixed point of~$f$, which in our case is $\sigma^*$. 

\begin{algorithm}
\SetAlgoLined
\SetKwInOut{Input}{Input}
\SetKwInOut{Output}{Output}
\Input{Function $f$ given in Equation~\eqref{eq:function_f}, initial condition for privacy noise variance $\sigma_0$ given in Equation~\eqref{eq:sigma_0}}
\Output{Optimal privacy parameter~$(\sigma^*)^2$}
left$=\sigma_0^2$\;
right$=\frac{\left[(b-a)+\frac{\Delta Q}{2}\right]\Delta Q}{\epsilon-\ln(\Delta C(\sigma_0))}$\;
intervalSize$=($left$+$right$)/2$\;
\While{intervalSize$>$right$-$left}{
    intervalSize $=$ right $-$ left\;
    $(\sigma^*)^2=\frac{\text{left}+\text{right}}{2}$\;
    \If{$\frac{\left[(b-a)+\frac{\Delta Q}{2}\right]\Delta Q}{\epsilon-\ln(\Delta C(\sigma^*))}\geq(\sigma^*)^2$}{
        left $=(\sigma^*)^2$\;
    }
    \If{$\frac{\left[(b-a)+\frac{\Delta Q}{2}\right]\Delta Q}{\epsilon-\ln(\Delta C(\sigma^*))}\leq(\sigma^*)^2$}{
        right $=(\sigma^*)^2$ \;
    }
}
\Return $(\sigma^*)^2$
\caption{A Zero Finding Algorithm}
 \label{alg:zero_finding_algorithm}
\end{algorithm}

The output of Algorithm~\ref{alg:zero_finding_algorithm} can be used
in Definition~\ref{def:bounded_guassian} to form the univariate bounded Gaussian mechanism,
and Theorem~\ref{thm:main_result_1} shows that doing so provides~$\epsilon$-differential privacy.
We next present an analogous mechanism for multi-variate query responses.

\section{Multivariate Bounded Gaussian mechanism} \label{sec:multivariate_mechanism}
We begin by formally defining the multivariate bounded Gaussian mechanism, and it
will be the focus of the remainder of this section.

\begin{defi}[Multivariate Bounded Gaussian Mechanism]\label{def:bounded_guassian_nd}
    Fix~$\sigma_m>0$. Let a query~$Q : S^n \to D$ be given, where $D=[\mathbf{a},\mathbf{b}]\subset \mathbb{R}^m$.
    For a database $d\in S^n$ and its output $Q(d)=\mathbf{s}\in D$,  the multivariate bounded Gaussian mechanism $M_B^m:\Omega\times \mathbb{R}^n\to D$ is given by the probability density function
    \begin{equation}
    p_B(\mathbf{x})=\begin{cases}
        \frac{\exp{\left(-\frac{1}{2}(\mathbf{x}-\mathbf{s})^T\Sigma^{-1}(\mathbf{x}-\mathbf{s})\right)}}{\int_\mathbf{a}^\mathbf{b}\exp{\left(-\frac{1}{2}(\mathbf{x}-\mathbf{s})^T\Sigma^{-1}(\mathbf{x}-\mathbf{s})\right)}d\mathbf{x}} & \text{if } \mathbf{x}\in D,\\[5pt]
        0 & \text{otherwise,}
    \end{cases}
\end{equation}
where $\Sigma = \sigma_m^2I \in \mathbb{R}^{m \times m}$.
\end{defi}

Before introducing the this mechanism's differential privacy guarantee we first require some preliminary results.

\subsection{Preliminary Results}
\begin{lem}\label{lem:preli_result_3}
    Fix~$\sigma_m>0$. For a database $d\in S^n$ and a query $Q(d)=\mathbf{s}\in D$ and $D=[\mathbf{a},\mathbf{b}]\subset \mathbb{R}^m$, let
    \begin{align}
        C_m(\mathbf{s},\sigma_m) = \frac{1}{\int_\mathbf{a}^\mathbf{b}\exp{\left(-\frac{1}{2}(\mathbf{x}-\mathbf{s})^T\Sigma^{-1}(\mathbf{x}-\mathbf{s})\right)}d\mathbf{x}}.\label{eq:definition_C_multidim}
    \end{align}
    Let an adjacent database $d'\in S^m$ satisfy $Q(d')=\mathbf{s}'\in D$. Then we have
    \begin{align*}
        \max_{\mathbf{s},\mathbf{s}'} \frac{C_m(\mathbf{s},\sigma_m)}{C_m(\mathbf{s}',\sigma_m)} = \frac{C_m(\mathbf{a},\sigma_m)}{C_m(\mathbf{a}+\mathbf{c},\sigma_m)},
    \end{align*}
    where $c_i=|s_i-s_i'|$ for the $i^{th}$ element of $\mathbf{c}\in\mathbb{R}^m$ and $||\mathbf{c}||_2\leq\Delta Q$.
\end{lem}
\begin{proof}
    See Appendix~\ref{apdx:proof_preli_result_3}.
\end{proof}

Now we define $\Delta C_m(\sigma_m,\mathbf{c}^*)$, where
\begin{align*}
    \frac{C_m(\mathbf{a},\sigma_m)}{C_m(\mathbf{a}+\mathbf{c},\sigma_m)}\leq \frac{C_m(\mathbf{a},\sigma_m)}{C_m(\mathbf{a}+\mathbf{c}^*,\sigma_m)}=:\Delta C_m(\sigma_m,\mathbf{c}^*),
\end{align*}
where $\mathbf{c}^*\in\mathbb{R}^m$ is the optimal value of $\mathbf{c}=(c_1,\dots,c_m)^T$ in the following optimization problem:
\begin{align}
    \max_{\mathbf{c}\in\mathbb{R}^m}\,\,\,&\Delta C_m(\sigma_m,\mathbf{c})\label{eq:objective_function}\\
    \text{subject to } \,\, &0\leq c_i\leq b_i-a_i \text{ for each }i\nonumber\\
    &||\mathbf{c}||_2 \leq \Delta Q.\nonumber
\end{align}

\begin{rem}
The optimization problem in~\eqref{eq:objective_function} is a concave maximization problem. 
In general, the value of~$c^*$ does not have a closed form, but it can be efficiently
computed 
numerically using standard optimization software such as CVX or CasADi. 
\end{rem}

\subsection{Main Results}
We now proceed to the main result of this section, which bounds the variance 
required for the multivariate bounded Gaussian mechanism to provide~$\epsilon$-differential privacy. 
\begin{thm}\label{thm:main_result_2}
    Let~$\epsilon>0$ and a query~$Q : S^n \to D$ be given. 
    For a database $d\in S^n$ and its output $Q(d)=\mathbf{s}\in D$, where $D=[\mathbf{a},\mathbf{b}]\subset \mathbb{R}^n$, 
    the multivariate bounded Gaussian mechanism $M_B$ provides $\epsilon$-differential privacy if
    \begin{equation} \label{eq:sigma_m_multi}
        \sigma_m^2\geq \frac{\left[||\mathbf{b-a}||_2+\frac{\Delta Q}{2}\right]\Delta Q}{\epsilon-\ln(\Delta C_m(\sigma_m,\mathbf{c}^*))}.
    \end{equation}
\end{thm}

\begin{proof}
    See Appendix~\ref{apdx:proof_main_result_2}.
\end{proof}

As with the univariate bounded Gaussian mechanism, we see that~$\sigma_m$ appears on both sides of~\eqref{eq:sigma_m_multi}. Hence, we will use numerical
methods to compute it. 

\subsection{Calculating \texorpdfstring{$\sigma_m$}{Lg}}
We will follow similar steps to those in Section~\ref{sec:calculate_sigma}. Let $\sigma_m^*$ be the minimal admissible value of~$\sigma_m$ that satisfies~\eqref{eq:sigma_m_multi}. 
Then for $D\subset \mathbb{R}^m$ we have
\begin{equation}
    f_m(\sigma_m^*)=(\sigma_m^*)^2 - \frac{\left[||\mathbf{b}-\mathbf{a}||_2+\frac{\Delta Q}{2}\right]\Delta Q}{\epsilon-\ln(\Delta C_m(\sigma_m^*,\mathbf{c}^*))}=0.\label{eq:f_n}
\end{equation}

We now present some lemmas concerning the function $f_n$ with respect to the 
point 
\begin{equation}
    \sigma_{m,0}=\sqrt{\frac{\left[||\mathbf{b}-\mathbf{a}||_2+\frac{\Delta Q}{2}\right]\Delta Q}{\epsilon}}.\label{eq:sigma_0_nd}
\end{equation}

\begin{lem}\label{lem:denominator_greater_0_nd}
    For any $\epsilon>0$, we have $\epsilon-\ln(\Delta C_m(\sigma_{m,0},\mathbf{c}^*))>0$, where $\mathbf{c}^*$ is defined in Equation~\eqref{eq:objective_function}.
\end{lem}
\begin{proof}
See Appendix~\ref{apdx:proof_denominator_greater_0_nd}.
\end{proof}

\begin{lem}\label{lem:delta_c_greater_0_nd}
    For $\sigma_{m,0}=\sqrt{\frac{\left[||\mathbf{b}-\mathbf{a}||_2+\frac{\Delta Q}{2}\right]\Delta Q}{\epsilon}}$, we have $\ln(\Delta C_m(\sigma_{m,0},\mathbf{c}^*))>0$.
\end{lem}
\begin{proof}
See Appendix~\ref{apdx:proof_delta_c_greater_0_nd}.
\end{proof}

\begin{lem}\label{lem:lemma_f_sigma_0_nd}
    For a point $\sigma_{m,0}=\sqrt{\frac{\left[||\mathbf{b}-\mathbf{a}||_2+\frac{\Delta Q}{2}\right]\Delta Q}{\epsilon}}$, we have $f_m(\sigma_{m,0})<0$.
\end{lem}
\begin{proof}
By plugging in $\sigma_{m,0}$ we have
\begin{align*}
    f_m(\sigma_{m,0})=\frac{\left[||\mathbf{b}-\mathbf{a}||_2+\frac{\Delta Q}{2}\right]\Delta Q}{\epsilon}-\frac{\left[||\mathbf{b}-\mathbf{a}||_2+\frac{\Delta Q}{2}\right]\Delta Q}{\epsilon-\ln(\Delta C_m(\sigma_{m,0},\mathbf{c}^*))}.
\end{align*}
With Lemma~\ref{lem:denominator_greater_0_nd} and Lemma~\ref{lem:delta_c_greater_0_nd} we have $\epsilon>\epsilon-\ln(\Delta C_m(\sigma_{m,0},\mathbf{c}^*))>0$. Therefore $f_m(\sigma_{m,0})<0$.
\end{proof}

\begin{lem}\label{lem:lemma_f_decreasing_nd}
    For any $\sigma_m\in[\sigma_{m,0},\infty)$, we have $f_m'(\sigma)>0$.
\end{lem}
\begin{proof}
See Appendix~\ref{apdx:proof_lemma_f_decreasing_nd}.
\end{proof}
By combining Lemma~\ref{lem:lemma_f_sigma_0_nd} and Lemma~\ref{lem:lemma_f_decreasing_nd} we can use Algorithm 1 in~\cite{holohan2018bounded} to find the exact fixed point of $f_m$, namely $\sigma_m^*$, which is shown in Algorithm~\ref{alg:zero_finding_algorithm_2}.
The output of Algorithm~\ref{alg:zero_finding_algorithm_2} can be combined with
the mechanism in Definition~\ref{def:bounded_guassian_nd}, and 
Theorem~\ref{thm:main_result_2} shows that this combination gives~$\epsilon$-differential privacy
for multivariate queries. 

\begin{algorithm}
\SetAlgoLined
\SetKwInOut{Input}{Input}
\SetKwInOut{Output}{Output}
\Input{Function $f_m$ given in Equation~\eqref{eq:f_n}, Initial privacy parameter $\sigma_{m,0}$ given in Equation~\eqref{eq:sigma_0_nd}.}
\Output{Optimal privacy parameter~$(\sigma_m^*)^2$.}
left$=\sigma_{m,0}^2$\;
Compute $c^*$ from the optimization problem in~\eqref{eq:objective_function}\;
rignt$=\frac{\left[||\mathbf{b-a}||_2+\frac{\Delta Q}{2}\right]\Delta Q}{\epsilon-\ln(\Delta C_m(\sigma_{m,0},\mathbf{c}^*))}$\;
intervalSize$=($left$+$right$)*2$\;
\While{intervalSize$>$right$-$left}{
    intervalSize $=$ right $-$ left\;
    $(\sigma_m^*)^2=\frac{\text{left}+\text{right}}{2}$\;
    Compute $c^*$ from the optimization problem in \eqref{eq:objective_function}\;
    \If{$\frac{\left[||b-a||_2+\frac{\Delta Q}{2}\right]\Delta Q}{\epsilon-\ln(\Delta C_m(\sigma_m^*,\mathbf{c}^*))}\geq(\sigma_m^*)^2$}{
        left $=(\sigma_m^*)^2$\;
    }
    \If{$\frac{\left[||\mathbf{b-a}||+\frac{\Delta Q}{2}\right]\Delta Q}{\epsilon-\ln(\Delta C_m(\sigma_m^*,\mathbf{c}^*))}\leq(\sigma_m^*)^2$}{
        right $=(\sigma_m^*)^2$\;
    }
}
\Return $(\sigma_m^*)^2$
\caption{A Zero Finding Algorithm for Multivariate Bounded Gaussian Mechanism}
 \label{alg:zero_finding_algorithm_2}
\end{algorithm}

\section{Numerical Results}\label{S:algo}
This section presents simulation results. 
We consider queries of the properties of a graph. Specifically, the bounded data we consider is (i) the average degree in an undirected graph, and (ii)
the second-smallest eigenvalue of the Laplacian of an undirected graph. This eigenvalue is also called the Fiedler value~\cite{fiedler73} of the graph. It is known that an undirected graph
is connected if and only if its Fiedler value is positive~\cite{fiedler73}, 
and the Fiedler value also sets the rate of convergence of various 
dynamical processes over graphs~\cite{mesbahi10}. 

We let $G=(V,E)$ denote a connected, undirected graph on $10$ nodes. In this experiment we have a two-dimensional query $Q$ to compute (i) the algebraic connectivity, equal to the second smallest eigenvalue $\lambda_{2}$ of the Laplacian matrix of~$G$, which satisfies $\lambda_{2}\in[0,10]$ here, and (ii) the degree of a fixed but arbitrary node $i$,
denoted~$N^i$, which satisfies $N^i \in [1,9]$ here. 
Therefore, for a graph~$G$ we have $Q(G)=[\lambda_{2},N_{\text{out}}]^T\in D$, where $D= [\mathbf{a,b}]\subset\mathbf{R}^2$, with $\mathbf{a}=[0,1]^T$ and $\mathbf{b}=[10,9]^T$. For a fixed~$k \in \mathbb{N}$, we say that two graphs are adjacent if they have the same node set but differ in $k$ edges. 
Let $\Delta Q_{\lambda_2}$ denote the sensitivity of $\lambda_2$, let $\Delta Q_{N^i}$ denote the sensitivity of node $i$'s degree, and let $\Delta Q=||(\Delta Q_{\lambda_2},\Delta Q_{N^i})^T||_2$ denote the sensitivity of the query $Q$. From~\cite[Lemma 1]{chen2021} we have $\Delta Q_{\lambda_2}=2k$ and $\Delta Q_{N^i}=k$. 
Thus,~$\Delta Q = \sqrt{5}k$. 
In our simulations we set $k=2$.
Table~\ref{table:1} gives some example variances computed for using the multivariate
bounded Gaussian mechanism. In Figure~\ref{fig:Normal_distribution}, there is a general decrease in the variance of the bounded Gaussian mechanism as $\epsilon$ grows. This 
decrease 
agrees with intuition because a larger $\epsilon$ implies weaker privacy protection and thus the private output distribution has a higher peak on the true query answer. 

We now compare the proposed mechanism with generalized Gaussian mechanism~\cite{Liu2019}. Let $\sigma_{GG}^2$ be the variance of generalized Gaussian mechanism and $\sigma_{BG}^2$ be the bounded Gaussian mechanism. Then we define the Percent Reduction in variance as
\begin{equation*}
    \text{Percent Reduction} = \frac{\sigma_{GG}^2-\sigma_{BG}^2}{\sigma_{GG}^2}\times 100\%.
\end{equation*}

Both Figure~\ref{fig:Normal_distribution} and Table~\ref{table:1} show that the bounded Gaussian mechanism always generates smaller variance, i.e., $\sigma_{BG}^2<\sigma_{GG}^2$ and
the $\text{Percent Reduction}>0$ for all $\epsilon$. In other words, compared to the 
generalized Gaussian mechanism~\cite{Liu2019}, 
the bounded Gaussian mechanism generates private outputs with less noise but with the same level of privacy protection. 


\begin{table}
\begin{tabular}{|c|c|c|c|} 
 \hline
 $\epsilon$ & Generalized Gaussian $\sigma_{GG}^2$ & Bounded Gaussian $\sigma_{BG}^2$ & Percent Reduction \\ [0.5ex] 
 \hline\hline
 0.1 & 1320.0 & 857.5 & 35.0\% \\ 
 \hline
 0.5 & 264.0 & 170.3 & 35.5\% \\
 \hline
 1.0 & 132.0 & 84.3 & 36.1\% \\
 \hline
 1.5 & 88 & 55.8 & 36.6\% \\
 \hline
 2.0 & 66 & 41.5 & 37.2\% \\ 
 \hline
 2.5 & 52.8 & 32.9 & 37.7\% \\ 
 \hline
 3.0 & 44 & 27.2 & 38.2\% \\[1ex] 
 \hline
\end{tabular}
\caption{Some example values of $\sigma_{BG}$ with different values of $\epsilon$. The last column shows the percentage reduction in variance attained by using the bounded Gaussian mechanism we develop.} 
\label{table:1}
\end{table}

\begin{figure}[h]
    \centering
    \includegraphics{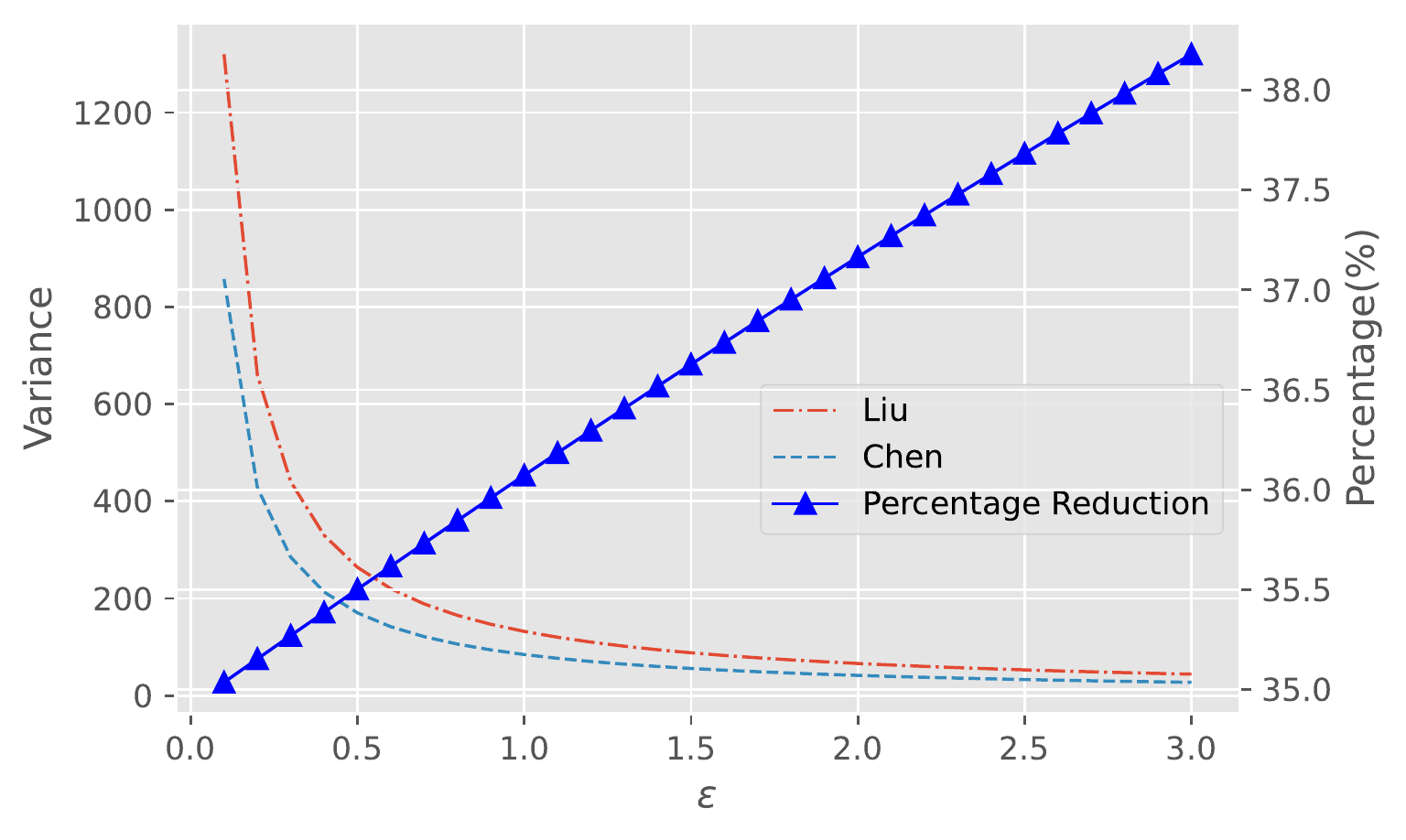}
    \caption{The variance comparison of proposed machanism and Liu's mechanism from~\cite[Definition 5]{Liu2019}. 
    With a larger $\epsilon$, which gives weaker privacy, the variance of bounded Gaussian mechanism decreases quickly and is always smaller than 
    that of the generalized Gaussian mechanism. This means that the bounded Gaussian mechanisms require less noise and provide better accuracy while maintaining the same level of protection. 
    In addition, the blue triangles ascent from left to right, indicating that the reduction in variance grows as~$\epsilon$ grows. 
    } 
    \label{fig:Normal_distribution}
\end{figure}

\section{Conclusion}\label{S:item}
This paper presented two differential privacy mechanisms, namely the univariate and multivariate bounded Gaussian mechanisms, 
for bounded domain queries of a database of sensitive data. Compared to the existing generalized Gaussian mechanism, the bounded Gaussian mechanisms we present generate private outputs with less noise
and better accuracy for the same privacy level. Future work will apply this mechanism to real world applications, such as privately forecasting 
epidemic propagation, federated learning and optimization, and explore privacy and performance trade-off. Besides, future work will design a denoise post-processing procedure to generate more accurate private outputs. 



\bibliography{sample}
\bibliographystyle{abbrvnat}

\appendix
\section{Proof of Lemma~\ref{lem:preli_result_1}}\label{apdx:proof_preli_result_1}
By symmetry we can write $s'=s+c$. We proceed in this proof by showing that $\frac{C(s,\sigma)}{C(s+c,\sigma)}$ is monotonically decreasing with respect to $s$. In other words, we show $\frac{\partial}{\partial s}\left(\frac{C(s,\sigma)}{C(s+c,\sigma)}\right)<0$. We first note that
\begin{align}
    \frac{C(s,\sigma)}{C(s+c,\sigma)} &= \frac{\Phi\left(\frac{b-s-c}{\sigma}\right)-\Phi\left(\frac{a-s-c}{\sigma}\right)}{\Phi\left(\frac{b-s}{\sigma}\right)-\Phi\left(\frac{a-s}{\sigma}\right)} \nonumber\\
    &= \frac{\int_a^b \phi\left(\frac{x-s-c}{\sigma}\right)dx}{\int_a^b \phi\left(\frac{x-s}{\sigma}\right)dx}\nonumber \\
    &= \frac{\int_a^b\exp{\left(-\frac{1}{2}\left(\frac{x-s-c}{\sigma}\right)^2\right)}dx}{\int_a^b\exp{\left(-\frac{1}{2}\left(\frac{x-s}{\sigma}\right)^2\right)}dx}\nonumber\\
    &= \frac{\int_{a-s}^{b-s}\exp{\left(-\frac{1}{2}\left(\frac{y-c}{\sigma}\right)^2\right)}dy}{\int_a^b\exp{\left(-\frac{1}{2}\left(\frac{x-s}{\sigma}\right)^2\right)}dx},\label{eq:1}
\end{align}
where $y=x-s$. We now introduce the Leibniz rule that we will use later in the proof. 
\begin{thm}[Leibniz Integral Rule]\label{thm:Leibniz_Rule}
    Let $g(x,t)$ be a function such that both $g(x,t)$ and its partial derivative $g_x(x,t)$ are continuous in $t$ and $x$ in some region of the $xt$-plane, including $f_1(x)\leq t\leq f_2(x)$, $x_0\leq x\leq x_1$. Also suppose that the functions $f_1(x)$ and $f_2(x)$ are both continuous and both have continuous derivatives for $x_0\leq x\leq x_1$. Then for $x_0\leq x\leq x_1$,
    \begin{equation*}
        \frac{d}{dx}\left(\int_{f_1(x)}^{f_2(x)}g(x,t)dt\right)=g(x,f_2(x))\cdot\frac{d}{dx}f_2(x) - g(x,f_1(x))\cdot\frac{d}{dx}f_1(x)+\int_{f_1(x)}^{f_2(x)}\frac{d}{dx}g(x,t)dt.
    \end{equation*}
\end{thm}

Let $A(s)=\int_a^b\exp{\left(-\frac{1}{2}\left(\frac{x-s}{\sigma}\right)^2\right)}dx$. Using Theorem~\ref{thm:Leibniz_Rule}, the derivative of Equation~\eqref{eq:1} becomes
\begin{align}
    \frac{\partial}{\partial s} \left(\frac{C(s,\sigma)}{C(s+c,\sigma)}\right)&=\frac{[\exp(-\frac{(a-c-s)^2}{2\sigma^2})-\exp(-\frac{(b-c-s)^2}{2\sigma^2})]A(s)}{A(s)^2}  \nonumber\\
    & \quad -\frac{[\exp(-\frac{(a-s)^2}{2\sigma^2})-\exp(-\frac{(b-s)^2}{2\sigma^2})]\int_a^b\exp(-\frac{(x-c-s)^2}{2\sigma^2})dx}{A(s)^2}.\label{eq:2}
\end{align}

We next introduce the mean value theorem of integrals.
\begin{thm}[Mean Value Theorem of Integrals]\label{thm:mean_value_thm}
    For a continuous function $h:(a,b)\to \mathbb{R}$ and an integrable function $g$ that does not change sign on $[a,b]$, there exists $\alpha\in[a,b]$ such that
    \begin{equation*}
        \int_a^bh(x)g(x)dx = h(\alpha)\int_a^bg(x)dx.
    \end{equation*}
\end{thm}

Using Theorem~\ref{thm:mean_value_thm}, we have
\begin{align}
    \int_a^b\exp\left(-\frac{(x-c-s)^2}{2\sigma^2}\right)dx &= \int_a^b\exp\left(-\frac{1}{2\sigma^2}((x-s)^2-2(x-s)c+c^2)\right)dx \nonumber\\
    &= \exp\left(-\frac{c^2}{2\sigma^2}\right)\int_a^b \exp\left(-\frac{(x-s)^2}{2\sigma^2}\right)\exp\left(\frac{(x-s)c}{\sigma^2}\right)dx \nonumber\\
    &= \exp\left(-\frac{c^2}{2\sigma^2}\right)\exp\left(\frac{(e-s)c}{\sigma^2}\right)A(s),\label{eq:3}
\end{align}
where $e\in[a,b]$. Then we plug Equation~\eqref{eq:3} into the Equation~\eqref{eq:2}, and we have
\begin{align}
   &\frac{\partial}{\partial s} \left(\frac{C(s,\sigma)}{C(s+c,\sigma)}\right)\\
   &=\exp\left(-\frac{c^2}{2\sigma^2}\right)\frac{[\exp(-\frac{(a-s)^2}{2\sigma^2})\exp(\frac{(a-s)c}{\sigma^2})-\exp(-\frac{(b-s)^2}{2\sigma^2})\exp(\frac{(b-s)c}{\sigma^2})]}{A(s)} \nonumber\\
    &\quad- \exp\left(-\frac{c^2}{2\sigma^2}\right)\frac{[\exp(-\frac{(a-s)^2}{2\sigma^2})-\exp(-\frac{(b-s)^2}{2\sigma^2})]\exp(\frac{(e-s)c}{\sigma^2})}{A(s)} \label{eq:4}\\ &=\frac{1}{A(s)}\exp\left(-\frac{c^2}{2\sigma^2}\right)\left[\exp\left(-\frac{(a-s)^2}{2\sigma^2}\right)\left(\exp\left(\frac{(a-s)c}{\sigma^2}\right)-\exp\left(\frac{(e-s)c}{\sigma^2}\right)\right)\right] \nonumber\\
    &\quad + \frac{1}{A(s)}\exp\left(-\frac{c^2}{2\sigma^2}\right)\left[\exp\left(-\frac{(b-s)^2}{2\sigma^2}\right)\left(\exp\left(\frac{(e-s)c}{\sigma^2}\right)-\exp\left(\frac{(b-s)c}{\sigma^2}\right)\right)\right]\label{eq:5}\\
    &= \frac{1}{A(s)}\exp\left(-\frac{c^2}{2\sigma^2}\right)\left[\exp\left(-\frac{(a-s)^2}{2\sigma^2}\right)\cdot B(s,c) + \exp\left(-\frac{(b-s)^2}{2\sigma^2}\right)\cdot C(s,c)\right],\label{eq:6}
\end{align}

where
\begin{align*}
    B(s,c)= \exp\left(\frac{(a-s)c}{\sigma^2}\right)-\exp\left(\frac{(e-s)c}{\sigma^2}\right) \leq 0, \quad \text{and equality is achieved if and only if } e=a,\\
    C(s,c)= \exp\left(\frac{(e-s)c}{\sigma^2}\right)-\exp\left(\frac{(b-s)c}{\sigma^2}\right) \leq 0, \quad \text{if equality is achieved if and only if } e=b.
\end{align*}

The Equation~\eqref{eq:4} comes by rewriting equation~\ref{eq:2} using perfect square trinomial. The Equation~\eqref{eq:5} comes by rearranging these terms and factors out their common factors. The Equation~\eqref{eq:6} is just substituting $B(s,c)$ and $C(s,c)$.

Since $B(s,c)$ and $C(s,c)$ cannot be 0 at the same time, we have 
\begin{equation*}
    \frac{\partial}{\partial s} \left(\frac{C(s,\sigma)}{C(s+c,\sigma)}\right)< 0.
\end{equation*}

Since $s\in[a,b]$, we have
\begin{equation*}
    \max_s \frac{C(s,\sigma)}{C(s+c,\sigma)} = \frac{C(a,\sigma)}{C(a+c,\sigma)}.
\end{equation*}

\section{Proof of Lemma~\ref{lem:preli_result_2}}\label{apdx:proof_preli_result_2}
We proceed this proof by forming a maximization problem and then use the Lagrangian multipliers method to solve it. By symmetry we assume $c>0$ and we have
\begin{align*}
    &\max_c \frac{C(a,\sigma)}{C(a+c,\sigma)}=\max_c\frac{\int_a^b\exp{\left(-\frac{1}{2}\left(\frac{x-a-c}{\sigma}\right)^2\right)}dx}{\int_a^b\exp{\left(-\frac{1}{2}\left(\frac{x-a}{\sigma}\right)^2\right)}dx} \\
    &\text{subject to} \quad 
     c\in [0, \Delta Q].
\end{align*}
Because this is a maximization problem, the Lagrangian function is formed as
\begin{equation*}
    \mathcal{L}(c,\lambda_1,\lambda_2) = \frac{C(a,\sigma)}{C(a+c,\sigma)} - \lambda_1(0-c) - \lambda_2(c-\Delta Q),
\end{equation*}
where $\lambda_1$,$\lambda_2\geq 0$. We then set $\frac{\partial\mathcal{L}(c,\lambda_1,\lambda_2)}{\partial c}=0$ and apply the Karush–Kuhn–Tucker (KKT) conditions to find  

\begin{align}
    &\frac{\partial}{\partial c}\left(\frac{C(a,\sigma)}{C(a+c,\sigma)}\right) +\lambda_1 - \lambda_2 = 0\label{eq:langrage_1} \\
    &\lambda_1(0-c) = 0 \label{eq:langrage_2}\\
    &\lambda_2(c-\Delta Q) = 0 \label{eq:langrage_3}
\end{align}

The derivative term in Equation~\eqref{eq:langrage_1} can be further simplified as
\begin{align}
    \frac{\partial}{\partial c}\left(\frac{C(a,\sigma)}{C(a+c,\sigma)}\right) &= \frac{1}{A(s)}\left(\frac{\partial}{\partial c} \int_a^b\exp{\left(-\frac{1}{2}\left(\frac{x-a-c}{\sigma}\right)^2\right)}dx\right) \nonumber\\
    &= \frac{1}{A(s)}\left(\frac{\partial}{\partial c} \int_{a-c}^{b-c}\exp{\left(-\frac{1}{2}\left(\frac{z-a}{\sigma}\right)^2\right)}dz\right) \nonumber\\
    &=\frac{1}{A(s)}\left(\exp\left(-\frac{(a-a-c)^2}{2\sigma^2}\right)-\exp\left(-\frac{(b-a-c)^2}{2\sigma^2}\right)\right) \label{eq:7}\\
    &= \frac{1}{A(s)}\left(\exp\left(-\frac{c^2}{2\sigma^2}\right)-\exp\left(-\frac{(c-(b-a))^2}{2\sigma^2}\right)\right),\label{eq:8}
\end{align}

where $z=x-c$. Equation~\eqref{eq:7} comes by using Theorem~\ref{thm:Leibniz_Rule}. We now consider two cases.

\emph{Case 1: \texorpdfstring{$\Delta Q< \frac{b-a}{2}$.}{Lg}}

Now if $c\leq\Delta Q<\frac{b-a}{2}$, then
\begin{align*}
    &0 \leq c < \frac{b-a}{2} \\
    &0\leq 2c < b-a \\
    &0\leq c < (b-a) - c \\
    &0\leq c^2 < (c-(b-a))^2\\
    &\exp\left(-\frac{c^2}{2\sigma^2}\right) > \exp\left(-\frac{(c-(b-a))^2}{2\sigma^2}\right).
\end{align*}

Therefore based on Equation~\eqref{eq:8} we have $\frac{\partial}{\partial c}\left(\frac{C(a,\sigma)}{C(a+c,\sigma)}\right) > 0$. Therefore Equations~\eqref{eq:langrage_1}-~\eqref{eq:langrage_3} are satisfied simultaneously if and only if $\lambda_2> 0$ and $\lambda_1=0$. Therefore $c=\Delta Q$. 

\emph{Case 2:  \texorpdfstring{$\frac{b-a}{2}\leq\Delta Q\leq b-a$}{Lg}}

Now if $\frac{b-a}{2}\leq\Delta Q\leq b-a$, then Equation~\eqref{eq:langrage_1}-~\eqref{eq:langrage_3} are satisfied simultaneously if and only if $c=\frac{b-a}{2}$, $\lambda_1=0$ and $\lambda_2=0$. 

Thus, we have 
\begin{align}
    c = \begin{cases}
    \Delta Q , &\text{ if } \Delta Q\leq \frac{b-a}{2},\\
    \frac{b-a}{2}, &\text{ Otherwise.}
    \end{cases}\label{eq:apdx_c}
\end{align}

\section{Proof of Theorem~\ref{thm:main_result_1}}\label{apdx:proof_main_result_1}
Let $d\in S^n$ and $d'\in S^n$ be adjacent, and let a query $Q:S^n\rightarrow D$ be given, where $D=[a,b]\subset\mathbb{R}$ with $|a|,|b|<\infty$. Let $Q(d)=s\in D$ and $Q(d')=s'\in D$ such that $s'=s+c$ and by Definition~\ref{def:sensitivity} we have $|c|\leq \Delta Q$. 
Without loss of generality we take $c \geq 0$. To prove differential privacy we examine the ratio of the probabilities that the bounded Gaussian mechanism outputs some element $z\in D$ on $s$ and $s'$, namely
\begin{align*}
    \frac{\Pr[M_B(s)=z]}{\Pr[M_B(s')=z]}&= \frac{\phi\left(\frac{z-s}{\sigma}\right)}{\phi\left(\frac{z-s'}{\sigma}\right)}\cdot\frac{\Phi\left(\frac{b-s'}{\sigma}\right)-\Phi\left(\frac{a-s'}{\sigma}\right)}{\Phi\left(\frac{b-s}{\sigma}\right)-\Phi\left(\frac{a-s}{\sigma}\right)} \nonumber\\
    &= \frac{\phi\left(\frac{z-s}{\sigma}\right)}{\phi\left(\frac{z-s-c}{\sigma}\right)}\cdot\frac{\Phi\left(\frac{b-s-c}{\sigma}\right)-\Phi\left(\frac{a-s-c}{\sigma}\right)}{\Phi\left(\frac{b-s}{\sigma}\right)-\Phi\left(\frac{a-s}{\sigma}\right)} \nonumber\\
    &= \frac{\phi\left(\frac{z-s}{\sigma}\right)}{\phi\left(\frac{z-s-c}{\sigma}\right)}\cdot \frac{C(s,\sigma)}{C(s+c,\sigma)}\\
    &\leq \Delta C(\sigma)\cdot |g(z,s,c|\sigma)|,
\end{align*}

where
\begin{align*}
    |g(z,s,c|\sigma)| &=\left|\frac{\phi\left(\frac{z-s}{\sigma}\right)}{\phi\left(\frac{z-s-c}{\sigma}\right)}\right|\\ &=\left|\frac{\exp\left(-\frac{x^2}{2\sigma^2}\right)}{\exp\left(-\frac{(x+c)^2}{2\sigma^2}\right)}\right|\\
    &= \left|\exp\left(\frac{(2xc+c^2)}{2\sigma^2}\right)\right|\\
    &\leq \left|\exp\left(\frac{(2x\Delta Q+(\Delta Q)^2)}{2\sigma^2}\right)\right|,
\end{align*}

with $x=s-z$. As in~\cite{Dwork2014}, for $\epsilon'>0$, if $x<\frac{\sigma^2\epsilon'}{\Delta Q}-\frac{\Delta Q}{2}$ then we have
\begin{equation*}
    |g(z,s,c|\sigma)| \leq \exp(\epsilon').
\end{equation*}

Since $x\in[0,b-a]$, we let
\begin{align*}
    \frac{\sigma^2\epsilon'}{\Delta Q}-\frac{\Delta Q}{2}\geq b-a.
\end{align*}

solving for $\sigma^2$ gives
\begin{equation*}
    \sigma^2 \geq \frac{\left[(b-a)+\frac{\Delta Q}{2}\right]\Delta Q}{\epsilon'}.
\end{equation*}

Now we let
\begin{align*}
    \frac{\Pr[M_B(s)=z]}{\Pr[M_B(s')=z]}&\leq \Delta C(\sigma)\cdot |g(z,s,c|\sigma)|\\
    &\leq  \Delta C(\sigma)\exp(\epsilon').
\end{align*}

To satisfy the differential privacy guarantee, we let $\exp(\epsilon)=\Delta C(\sigma)\exp(\epsilon')$. Then $\epsilon'=\epsilon-\ln{\Delta C(\sigma)}$. As a result, the scale parameter $\sigma$ has to satify $\sigma^2\geq \frac{\left[(b-a)+\frac{\Delta Q}{2}\right]\Delta Q}{\epsilon-\ln(\Delta C(\sigma))}$.

\section{Proof of Lemma~\ref{lem:denominator_greater_0}}\label{apdx:proof_denominator_greater_0}
By definition of $\sigma_0$ we have
    \begin{equation*}
        \Delta C(\sigma_0) = \begin{cases}
        \frac{C(a,\sigma_0)}{C(a+\Delta Q,\sigma_0)} & \text{if } \Delta Q\leq\frac{b-a}{2}\\[4pt]
        \frac{C(a,\sigma_0)}{C\left(\frac{b+a}{2},\sigma_0\right)} & \text{otherwise}.
        \end{cases}
    \end{equation*}

    We first consider $\Delta Q\leq \frac{b-a}{2}$. First note that $\Delta C(\sigma_0)$ can be further simplified as 
    \begin{align}
        \Delta C(\sigma_0) &= \frac{\int_a^b\exp{\left(-\frac{1}{2}\left(\frac{x-a-\Delta Q}{\sigma_0}\right)^2\right)}dx}{\int_a^b\exp{\left(-\frac{1}{2}\left(\frac{x-a}{\sigma_0}\right)^2\right)}dx}\nonumber\\
        &= \frac{\int_a^b\exp{\left(-\frac{1}{2}\left(\frac{x-a}{\sigma_0}\right)^2\right)}\exp\left(\frac{1}{2}\left(\frac{2(x-a)\Delta Q}{\sigma_0^2}\right)\right)\exp\left(-\frac{1}{2}\frac{(\Delta Q)^2}{\sigma_0^2}\right)dx}{\int_a^b\exp{\left(-\frac{1}{2}\left(\frac{x-a}{\sigma_0}\right)^2\right)}dx}\nonumber\\
        &= \frac{\exp\left(\frac{1}{2}\left(\frac{2(\delta-a)\Delta Q}{\sigma_0^2}\right)\right)\exp\left(-\frac{1}{2}\frac{(\Delta Q)^2}{\sigma_0^2}\right)\int_a^b\exp{\left(-\frac{1}{2}\left(\frac{x-a}{\sigma_0}\right)^2\right)}dx}{\int_a^b\exp{\left(-\frac{1}{2}\left(\frac{x-a}{\sigma_0}\right)^2\right)}dx}, \quad\text{ where } \delta\in[a,b]\nonumber\\
        &= \exp\left(\frac{1}{2}\left(\frac{2(\delta-a)\Delta Q}{\sigma_0^2}\right)\right)\exp\left(-\frac{1}{2}\frac{(\Delta Q)^2}{\sigma_0^2}\right).
        \label{eq:proof_lemma_1_2}
    \end{align}

The second equality holds by binomial theorem, and the third equality holds by using Theorem~\ref{thm:mean_value_thm}. We can now plug in $\sigma_0=\sqrt{\frac{\left[(b-a)+\frac{\Delta Q}{2}\right]\Delta Q}{\epsilon}}$ and get:
\begin{align*}
    \Delta C(\sigma_0)&=\exp\left(\frac{1}{2}\left(\frac{2(\delta-a)\Delta Q\epsilon}{\left(b-a+\frac{\Delta Q}{2}\right)\Delta Q}\right)\right)\exp\left(-\frac{1}{2}\left(\frac{(\Delta Q)^2\epsilon}{\left(b-a+\frac{\Delta Q}{2}\right)\Delta Q}\right)\right)\\
    &= \exp\left(\frac{1}{2}\left(
    \frac{(2(\delta-a)-\Delta Q)\Delta Q\epsilon}{(b-a+\frac{\Delta Q}{2})\Delta Q}\right)\right)\\
    &= \exp\left(\frac{(2(\delta-a)-\Delta Q)\Delta Q\epsilon}{(2(b-a)+\Delta Q)\Delta Q}\right)\\
    &< \exp(\epsilon).
\end{align*}

The last inequality holds since $\delta\in[a,b]$ and $2(\delta-a)-\Delta Q\leq 2(b-a)+\Delta Q$.

Now consider $\Delta Q> \frac{b-a}{2}$, then we plug $\sigma_0$ in and have
\begin{align*}
    \Delta C(\sigma_0) &= \frac{\int_a^b\exp{\left(-\frac{1}{2}\left(\frac{x-a-\frac{b-a}{2}}{\sigma_0}\right)^2\right)}dx}{\int_a^b\exp{\left(-\frac{1}{2}\left(\frac{x-a}{\sigma_0}\right)^2\right)}dx}\\
    &= \exp\left(\frac{1}{2}\left(\frac{2(\delta-a)\frac{b-a}{2}}{\sigma_0^2}\right)\right)\exp\left(-\frac{1}{2}\frac{\left(\frac{b-a}{2}\right)^2}{\sigma_0^2}\right)\\
    &=\exp\left(\frac{\left(2(\delta-a)-\frac{b-a}{2}\right)\frac{b-a}{2}}{(2(b-a)+\Delta Q)\Delta Q}\epsilon\right)\\
    &\leq \exp\left(\frac{\left(2(\delta-a)-\frac{b-a}{2}\right)\Delta Q}{(2(b-a)+\Delta Q)\Delta Q}\epsilon\right)\\
    &\leq \exp\left(\frac{\left(2(\delta-a)-\frac{b-a}{2}\right)}{(2(b-a)+\Delta Q)}\epsilon\right)\\
    &\leq \exp(\epsilon).
\end{align*}
The second equality is achieved by the same technique as Equation~\eqref{eq:proof_lemma_1_2}. The first inequality holds since $\frac{b-a}{2}\leq \Delta Q$, and the third inequality holds since $2(\delta-a)-\frac{b-a}{2}<2(b-a)+\Delta Q$.

\section{Proof of Lemma~\ref{lem:delta_c_greater_0}}\label{apdx:proof_delta_c_greater_0}
We prove the above inequality by examining the term $\Delta C(\sigma)$. According to Definition~\ref{def:sensitivity} and Equation~\eqref{eq:apdx_c} we have $c\in\left(0,\frac{b-a}{2}\right]$. Then

\begin{align*}
    \Delta C(\sigma)&= \frac{\int_a^b\exp{\left(-\frac{1}{2}\left(\frac{x-a-c}{\sigma}\right)^2\right)}dx}{\int_a^b\exp{\left(-\frac{1}{2}\left(\frac{x-a}{\sigma}\right)^2\right)}dx}\\
    &= \frac{\int_{a-c}^{b-c}\exp{\left(-\frac{1}{2}\left(\frac{x-a}{\sigma}\right)^2\right)}dx}{\int_a^b\exp{\left(-\frac{1}{2}\left(\frac{x-a}{\sigma}\right)^2\right)}dx}\\
    &= \frac{\int_{a-c}^{a}\exp{\left(-\frac{1}{2}\left(\frac{x-a}{\sigma}\right)^2\right)}dx+\int_{a}^{b-c}\exp{\left(-\frac{1}{2}\left(\frac{x-a}{\sigma}\right)^2\right)}dx}{\int_{b-c}^{b}\exp{\left(-\frac{1}{2}\left(\frac{x-a}{\sigma}\right)^2\right)}dx+\int_{a}^{b-c}\exp{\left(-\frac{1}{2}\left(\frac{x-a}{\sigma}\right)^2\right)}dx}\\
    &= \frac{\int_{-c}^{0}\exp{\left(-\frac{1}{2}\left(\frac{y}{\sigma}\right)^2\right)}dy+\int_{a}^{b-c}\exp{\left(-\frac{1}{2}\left(\frac{x-a}{\sigma}\right)^2\right)}dx}{\int_{-c}^{0}\exp{\left(-\frac{1}{2}\left(\frac{z+b-a}{\sigma}\right)^2\right)}dz+\int_{a}^{b-c}\exp{\left(-\frac{1}{2}\left(\frac{x-a}{\sigma}\right)^2\right)}dx},
\end{align*}

where $y=x-a$ and $z=x-b$. To proceed in the proof we introduce the comparison property of definite integrals.
\begin{thm}[Comparison Property of Definite Integrals]
    If $f(m)\geq g(m)$ for $\underline{m}\leq m\leq \bar{m}$, then
    \begin{equation*}
        \int_{\underline{m}}^{\bar{m}}f(m)dm\geq \int_{\underline{m}}^{\bar{m}}g(m)dm.
    \end{equation*}
\end{thm}

We now have 
\begin{equation*}
    \exp{\left(-\frac{1}{2}\left(\frac{y}{\sigma}\right)^2\right)}\geq \exp{\left(-\frac{1}{2}\left(\frac{z+b-a}{\sigma}\right)^2\right)}
\end{equation*}
on $y\in[-c,0]$, $z\in[-c,0]$ and $c\in\left(0,\frac{b-a}{2}\right]$. Therefore 
\begin{equation*}
    \int_{-c}^{0}\exp{\left(-\frac{1}{2}\left(\frac{y}{\sigma}\right)^2\right)}dy\geq \int_{-c}^{0}\exp{\left(-\frac{1}{2}\left(\frac{z+b-a}{\sigma}\right)^2\right)}dz.
\end{equation*}

Therefore the numerator of $\Delta C(\sigma)$ is greater than its denominator. As a result, $\Delta C(\sigma)>1$ and $\ln(\Delta C(\sigma))>0$.

\section{Proof of Lemma~\ref{lem:lemma_f_decreasing}}\label{apdx:proof_lemma_f_decreasing}
We take the derivative of $f(\sigma)$ and have
\begin{align*}
    f'(\sigma)&=2\sigma - \left(\frac{\left[(b-a)+\frac{\Delta Q}{2}\right]\Delta Q}{\epsilon-\ln(\Delta C(\sigma))}\right)'\\
    &= 2\sigma - \frac{\left[(b-a)+\frac{\Delta Q}{2}\right]\Delta Q}{(\epsilon-\ln(\Delta C(\sigma))^2}\cdot \frac{1}{\Delta C(\sigma)} \cdot (\Delta C(\sigma))'\\
    &= 2\sigma + h(\sigma),
\end{align*}
where
\begin{equation*}
    h(\sigma) = - \frac{\left[(b-a)+\frac{\Delta Q}{2}\right]\Delta Q}{(\epsilon-\ln(\Delta C(\sigma))^2}\cdot \frac{1}{\Delta C(\sigma)} \cdot (\Delta C(\sigma))'.
\end{equation*}

According to Lemma~\ref{lem:delta_c_greater_0} we have $\Delta C(\sigma)>1$ for all $\sigma\in(0,\infty]$. We next prove $(\Delta C(\sigma))'<0$ since $(\Delta C(\sigma))'<0$ it will follow that $h(\sigma)>0$.

We now examine $(\Delta C(\sigma))'$. By taking the derivative and using Theorem~\ref{thm:Leibniz_Rule} we have
\begin{align*}
    (\Delta C(\sigma))'&= \frac{d}{d\sigma}\left(\frac{\int_a^b\exp\left(-\frac{1}{2}\left(\frac{x-a-c}{\sigma}\right)^2\right)dx}{\int_a^b\exp\left(-\frac{1}{2}\left(\frac{y-a}{\sigma}\right)^2\right)dy}\right)\\
    &= \frac{d}{d\sigma}\left(\frac{\int_{\frac{-c}{\sigma}}^{\frac{b-a-c}{\sigma}}\exp\left(-\frac{1}{2}x^2\right)dx}{\int_{0}^{\frac{b-a}{\sigma}}\exp\left(-\frac{1}{2}y^2\right)dy}\right)\\
    &= \frac{\left[\exp\left(-\frac{1}{2}\left(\frac{b-a-c}{\sigma}\right)^2\right)\cdot \frac{-(b-a-c)}{\sigma^2}-\exp\left(-\frac{1}{2}\left(\frac{-c}{\sigma}\right)^2\right)\cdot \frac{c}{\sigma^2}\right]\cdot \left(\int_a^b\exp\left(-\frac{1}{2}\left(\frac{y-a}{\sigma}\right)^2\right)dy\right)}{\left(\int_a^b\exp\left(-\frac{1}{2}\left(\frac{y-a}{\sigma}\right)^2\right)dy\right)^2}\\
    &\quad - \frac{\left[\exp\left(-\frac{1}{2}\left(\frac{b-a}{\sigma}\right)^2\right)\cdot\frac{-(b-a)}{\sigma^2}\right]\cdot\int_a^b\exp\left(-\frac{1}{2}\left(\frac{x-a-c}{\sigma}\right)^2\right)dx }{\left(\int_a^b\exp\left(-\frac{1}{2}\left(\frac{y-a}{\sigma}\right)^2\right)dy\right)^2}\\
    &= \beta_1(\sigma) + \beta_2(\sigma), 
\end{align*}
where
\begin{align*}
    \beta_1(\sigma) &= \frac{\left[-\exp\left(-\frac{1}{2}\left(-\frac{b-a-c}{\sigma}\right)^2\right)\cdot \frac{b-a-c}{\sigma^2}-\exp\left(-\frac{1}{2}\left(\frac{c}{\sigma}\right)^2\right)\cdot \frac{c}{\sigma^2}\right]}{\int_a^b\exp\left(-\frac{1}{2}\left(\frac{y-a}{\sigma}\right)^2\right)dy}\\
    &= \frac{\frac{1}{\sigma^2}\exp\left(-\frac{c^2}{2\sigma^2}\right)\cdot\left[-\exp\left(-\frac{(b-a)^2-2(b-a)c}{2\sigma^2}\right)\cdot(b-a-c)-c\right]}{\int_a^b\exp\left(-\frac{1}{2}\left(\frac{y-a}{\sigma}\right)^2\right)dy}. 
\end{align*}

The second equality holds by splitting the square terms and reorganizing them. Moreover, we have
\begin{align*}
    \beta_2(\sigma) &= -\frac{\left[\exp\left(-\frac{1}{2}\left(\frac{b-a}{\sigma}\right)^2\right)\cdot\frac{-(b-a)}{\sigma^2}\right]\cdot \int_a^b\exp\left(-\frac{1}{2}\left(\frac{(x-a)^2-2(x-a)c+c^2}{\sigma^2}\right)\right)dx}{\left(\int_a^b\exp\left(-\frac{1}{2}\left(\frac{y-a}{\sigma}\right)^2\right)dy\right)^2}\\
    &= -\frac{\frac{1}{\sigma^2}\exp\left(-\frac{c^2}{2\sigma^2}\right)\left[-\exp\left(-\frac{1}{2}\left(\frac{b-a}{\sigma}\right)^2\right)(b-a)\cdot\exp\left(\frac{(\delta-a)c}{\sigma^2}\right)\cdot\int_a^b\exp\left(-\frac{(x-a)^2}{2\sigma^2}\right)dx\right]}{\left(\int_a^b\exp\left(-\frac{1}{2}\left(\frac{y-a}{\sigma}\right)^2\right)dy\right)^2}\\
    &= \frac{\frac{1}{\sigma^2}\exp\left(-\frac{c^2}{2\sigma^2}\right)\left[\exp\left(-\frac{1}{2}\left(\frac{b-a}{\sigma}\right)^2\right)(b-a)\cdot\exp\left(\frac{(\delta-a)c}{\sigma^2}\right)\right]}{\int_a^b\exp\left(-\frac{1}{2}\left(\frac{y-a}{\sigma}\right)^2\right)dy}.
\end{align*}

The second equality holds by applying Theorem~\ref{thm:mean_value_thm}, and the third equality holds by cancelling the integral term. Note $\delta\in(a,b)$. Then we have
\begin{align*}
    \beta_1(\sigma) + \beta_2(\sigma) &= \frac{\frac{1}{\sigma^2}\exp\left(-\frac{c^2}{2\sigma^2}\right)\cdot}{\int_a^b\exp\left(-\frac{1}{2}\left(\frac{y-a}{\sigma}\right)^2\right)dy}\\
    &\quad\left[c\left(\exp\left(-\frac{(b-a)^2-2(b-a)c}{2\sigma^2}\right)-1\right)+\right.\\
    &\quad\left.(b-a)\cdot\left[\exp\left(-\frac{(b-a)^2-2(\delta-a)c}{2\sigma^2}\right)-\exp\left(-\frac{(b-a)^2-2(b-a)c}{2\sigma^2}\right)\right]\right]\\
    &= \frac{\frac{1}{\sigma^2}\exp\left(-\frac{c^2}{2\sigma^2}\right)}{\int_a^b\exp\left(-\frac{1}{2}\left(\frac{y-a}{\sigma}\right)^2\right)dy}\cdot [\gamma_1(\sigma)+ \gamma_2(\sigma)],
\end{align*}
where
\begin{align*}
    \gamma_1(\sigma)&=c\left(\exp\left(-\frac{(b-a)^2-2(b-a)c}{2\sigma^2}\right)-1\right)\\
    &= c\left(\exp\left(-\frac{(b-a)(b-a-2c)}{2\sigma^2}\right)-1\right). 
\end{align*}
From Equation~\eqref{eq:apdx_c} we know
\begin{align*}
    c = \begin{cases}
    \Delta Q  &\text{ if } \Delta Q\leq \frac{b-a}{2}\\
    \frac{b-a}{2} &\text{ otherwise }. 
    \end{cases}
\end{align*}
Therefore $c\in\left(0,\frac{b-a}{2}\right]$. Then $\gamma_1(\sigma)\leq 0$. Now we examine $\gamma_2(\sigma)$:

\begin{align*}
    \gamma_2(\sigma)&=(b-a)\cdot\left[\exp\left(-\frac{(b-a)^2-2(\delta-a)c}{2\sigma^2}\right)-\exp\left(-\frac{(b-a)^2-2(b-a)c}{2\sigma^2}\right)\right]\\
    &<0.
\end{align*}
This inequality holds because $\delta\in(a,b)$. Therefore $v'(\sigma)=\beta_1(\sigma)+\beta_2(\sigma)<0$. And then we have $h(\sigma)>0$ and $f'(\sigma)>0$ on $[\sigma_0,\infty)$.

\section{Proof of Lemma~\ref{lem:preli_result_3}}\label{apdx:proof_preli_result_3}
By symmetry we assume $\mathbf{s}'=\mathbf{s}+\mathbf{c}$. We let $\mathbf{s}=(s_1,s_2,\dots,s_m)^T$. We proceed in this proof by showing that $\frac{C_m(\mathbf{s},\sigma_m)}{C_m(\mathbf{s}+\mathbf{c},\sigma_m)}$ is monotonically decreasing with respect to each $s_i$, $i\in\{1,2,\dots,m\}$. In other words, we show that $\frac{\partial}{\partial s_i}\left(\frac{C_m(\mathbf{s},\sigma_m)}{C_m(\mathbf{s+c},\sigma_m)}\right)<0$. We only show the proof of $\frac{\partial}{\partial s_1}\left(\frac{C_m(\mathbf{s},\sigma_m)}{C_m(\mathbf{s+c},\sigma_m)}\right)\leq 0$, but the proof for $s_2,\dots,s_m$ proceed in the same way. We first note that
\begin{align}
    \frac{C_m(\mathbf{s},\sigma_m)}{C_m(\mathbf{s+c},\sigma_m)}&= \frac{\int_\mathbf{a}^\mathbf{b}\exp{\left(-\frac{1}{2}(\mathbf{x}-\mathbf{s}-\mathbf{c})^T\Sigma^{-1}(\mathbf{x}-\mathbf{s}-\mathbf{c})\right)}dx}{\int_\mathbf{a}^\mathbf{b}\exp{\left(-\frac{1}{2}(\mathbf{x}-\mathbf{s})^T\Sigma^{-1}(\mathbf{x}-\mathbf{s})\right)}dx}\nonumber\\
    &=\frac{\int_{a_1}^{b_1}\int_{a_2}^{b_2}\cdots\int_{a_m}^{b_m}\exp\left(-\frac{1}{2\sigma_m^2}\left(\sum_{i=1}^n(x_i-s_i-c_i)^2\right)\right)dx_mdx_{m-1}\cdots dx_1}{\int_{a_1}^{b_1}\int_{a_2}^{b_2}\cdots\int_{a_m}^{b_m}\exp\left(-\frac{1}{2\sigma_m^2}\left(\sum_{i=1}^n(x_i-s_i)^2\right)\right)dx_mdx_{m-1}\cdots dx_1}\nonumber\\
    &= h_1(s_1,c_1)h_2(s_2,\dots,s_m,c_2,\dots,c_m),\label{eq:G_1}
\end{align}
where
\begin{align*}
    &h_1(s_1,c_1)=\frac{\int_{a_1}^{b_1}\exp\left(-\frac{1}{2\sigma_m^2}(x_1-s_1-c_1)^2\right)dx_1}{\int_{a_1}^{b_1}\exp\left(-\frac{1}{2\sigma_m^2}(x_1-s_1)^2\right)dx_1},\\
    &h_2(s_2,\dots,s_m,c_2,\dots,c_m)=\frac{\int_{a_2}^{b_2}\cdots\int_{a_m}^{b_m}\exp\left(-\frac{1}{2\sigma_m^2}\left(\sum_{i=2}^n(x_i-s_i-c_i)^2\right)\right)dx_mdx_{m-1}\cdots dx_2}{\int_{a_2}^{b_2}\cdots\int_{a_m}^{b_m}\exp\left(-\frac{1}{2\sigma_m^2}\left(\sum_{i=2}^n(x_i-s_i)^2\right)\right)dx_mdx_{m-1}\cdots dx_2}.
\end{align*}

Then take the derivative and we have
\begin{align*}
    \frac{\partial}{\partial s_1}\left(\frac{C_m(\mathbf{s},\sigma_m)}{C_m(\mathbf{s+c},\sigma_m)}\right)=\frac{\partial h_1(s_1,c_1)}{\partial s_1}\cdot h_2(s_2,\dots,s_m,c_2,\dots,c_m).
\end{align*}

Since $\frac{\partial h_2(\cdot)}{\partial s_1}=0$. We also have $h_2(\cdot)>0$, to show $\frac{\partial}{\partial s_1}\left(\frac{C_m(s,\sigma_m)}{C_m(s+c,\sigma_m)}\right)\leq 0$ we only need to show $\frac{\partial h_1(s_1,c_1)}{\partial s_1}\leq 0$, which can be proved by using the same technique as Appendix~\ref{apdx:proof_preli_result_1}.

\section{Proof of Theorem~\ref{thm:main_result_2}}\label{apdx:proof_main_result_2}
Fix two adjacent databases $d\in S^n$ and $d'\in S^n$, along with a query $Q$. 
Suppose that~$Q(d)=\mathbf{s}$ and $Q(d')=\mathbf{s}'$ such that $\mathbf{s'=s+c}$ and $||\mathbf{c}||_2\leq \Delta Q$. To prove 
that the multivariate mechanism provides
differential privacy, we examine
\begin{align*}
    \frac{\Pr[M_B^m(\mathbf{s})=\mathbf{z}]}{\Pr[M_B^m(\mathbf{s'})=\mathbf{z}]}&=\frac{\int_\mathbf{a}^\mathbf{b}\exp{\left(-\frac{1}{2}(\mathbf{x}-\mathbf{s}-\mathbf{c})^T\Sigma^{-1}(\mathbf{x}-\mathbf{s}-\mathbf{c})\right)}dx}{\int_\mathbf{a}^\mathbf{b}\exp{\left(-\frac{1}{2}(\mathbf{x}-\mathbf{s})^T\Sigma^{-1}(\mathbf{x}-\mathbf{s})\right)}dx}\\
    &\quad\cdot\frac{\exp{\left(-\frac{1}{2}(\mathbf{z}-\mathbf{s})^T\Sigma^{-1}(\mathbf{z}-\mathbf{s})\right)}}{\exp{\left(-\frac{1}{2}(\mathbf{z}-\mathbf{s}-\mathbf{c})^T\Sigma^{-1}(\mathbf{z}-\mathbf{s}-\mathbf{c})\right)}}\\
    &=  \frac{C_m(\mathbf{s},\sigma_m)}{C_m(\mathbf{s+c},\sigma_m)}\cdot\frac{\exp{\left(-\frac{1}{2}(\mathbf{z}-\mathbf{s})^T\Sigma^{-1}(\mathbf{z}-\mathbf{s})\right)}}{\exp{\left(-\frac{1}{2}(\mathbf{z}-\mathbf{s}-\mathbf{c})^T\Sigma^{-1}(\mathbf{z}-\mathbf{s}-\mathbf{c})\right)}}\\
    &\leq \Delta C_m(\sigma_m,\mathbf{c}^*)\cdot\left|\frac{\exp{\left(-\frac{1}{2}(\mathbf{z}-\mathbf{s})^T\Sigma^{-1}(\mathbf{z}-\mathbf{s})\right)}}{\exp{\left(-\frac{1}{2}(\mathbf{z}-\mathbf{s}-\mathbf{c})^T\Sigma^{-1}(\mathbf{z}-\mathbf{s}-\mathbf{c})\right)}}\right|. 
\end{align*}

Then, using the fact that~$\Sigma^{-1} = \frac{1}{\sigma_m^2}I$, we have 
\begin{align*}
    &\quad\left|\frac{\exp{\left(-\frac{1}{2}(\mathbf{z}-\mathbf{s})^T\Sigma^{-1}(\mathbf{z}-\mathbf{s})\right)}}{\exp{\left(-\frac{1}{2}(\mathbf{z}-\mathbf{s}-\mathbf{c})^T\Sigma^{-1}(\mathbf{z}-\mathbf{s}-\mathbf{c})\right)}}\right|\\ 
    &=\left|\frac{\exp\left(-\frac{||\mathbf{z-s}||_2^2}{2\sigma_m^2}\right)}{\exp\left(-\frac{||\mathbf{z-s-c}||_2^2}{2\sigma_m^2}\right)}\right|\\
    &= \left|\frac{\exp\left(-\frac{||\mathbf{x}||_2^2}{2\sigma_m^2}\right)}{\exp\left(-\frac{||\mathbf{x+c}||_2^2}{2\sigma_m^2}\right)}\right|\\
    &= \left|\exp\left(\frac{2\mathbf{x}^T\mathbf{c}+||\mathbf{c}||_2^2}{2\sigma_m^2}\right)\right|\\
    &\leq\left|\exp\left(\frac{2||\mathbf{x}||_2\Delta Q+(\Delta Q)^2}{2\sigma_m^2}\right)\right|,
\end{align*}
where we have set~$\mathbf{x=s-z}$. Based on Appendix A in~\cite{Dwork2014}, 
for $\epsilon'>0$, if $||\mathbf{x}||_2<\frac{\sigma_m^2\epsilon'}{\Delta Q}-\frac{\Delta Q}{2}$, 
then we have
\begin{equation}
    \left|\exp\left(\frac{2||\mathbf{x}||_2\Delta Q+(\Delta Q)^2}{2\sigma_m^2}\right)\right| \leq \exp(\epsilon').\label{eq:9}
\end{equation}

Since $||\mathbf{x}||_2\in\mathbf{\Big[0,||b-a||_2\Big]}$, then we choose the following inequality to ensure that Equation~\eqref{eq:9} always holds:
\begin{align}
    \frac{\sigma_m^2\epsilon'}{\Delta Q}-\frac{\Delta Q}{2}\geq ||\mathbf{b-a}||_2.\label{eq:10}
\end{align}
By rearranging Equation~\eqref{eq:10} we have a range for the scale parameter $\sigma_m$ 
\begin{equation*}
    \sigma_m^2 \geq \frac{\left[||\mathbf{b-a}||_2+\frac{\Delta Q}{2}\right]\Delta Q}{\epsilon'}.
\end{equation*}
Now we let
\begin{align*}
    \frac{\Pr[M_B^m(s)=z]}{\Pr[M_B^m(s')=z]}&\leq \Delta C_m(\sigma_m,\mathbf{c}^*)\cdot \left|\exp\left(\frac{2||\mathbf{x}||_2\Delta Q+(\Delta Q)^2}{2\sigma_m^2}\right)\right|\\
    &\leq  \Delta C_m(\sigma_m,\mathbf{c}^*)\exp(\epsilon').
\end{align*}
To satisfy the differential privacy guarantee, we let $\exp(\epsilon)=\Delta C_m(\sigma_m,\mathbf{c}^*)\exp(\epsilon')$. Then $\epsilon'=\epsilon-\ln{\Delta C_m(\sigma_m,\mathbf{c}^*)}$. As a result, the scale parameter $\sigma_m$ has to satify $\sigma_m^2\geq \frac{\left[||\mathbf{b-a}||_2+\frac{\Delta Q}{2}\right]\Delta Q}{\epsilon-\ln(\Delta C_m(\sigma_m,\mathbf{c}^*))}$.

\section{Proof of Lemma~\ref{lem:denominator_greater_0_nd}}\label{apdx:proof_denominator_greater_0_nd}
We proceed in this proof by showing that $\ln{(\Delta C_m(\sigma_{m,0},\mathbf{c}^*))}<\epsilon$. For $\sigma_{m,0}$ we have
\begin{align*}
    \Delta C_m(\sigma_{m,0},\mathbf{c}^*) &= \frac{C_m(\mathbf{a},\sigma_{m,0})}{C_m\left(\mathbf{a}+\mathbf{c}^*,\sigma_{m,0}\right)}\\
    &= \frac{\int_\mathbf{a}^\mathbf{b}\exp{\left(-\frac{1}{2}(\mathbf{x}-\mathbf{a}-\mathbf{c}^*)^T\Sigma_0^{-1}(\mathbf{x}-\mathbf{a}-\mathbf{c}^*)\right)}dx}{\int_\mathbf{a}^\mathbf{b}\exp{\left(-\frac{1}{2}(\mathbf{x}-\mathbf{a})^T\Sigma_0^{-1}(\mathbf{x}-\mathbf{a})\right)}dx}\\
    &= \frac{\exp\left(-\frac{1}{2}(\mathbf{c}^*)^T\Sigma_0^{-1}\mathbf{c}^*\right)\int_\mathbf{a}^\mathbf{b}\exp{\left(-\frac{1}{2}(\mathbf{x}-\mathbf{a})^T\Sigma_0^{-1}(\mathbf{x}-\mathbf{a})\right)}\cdot \exp((\mathbf{c}^*)^T\Sigma_0^{-1}(\mathbf{x-a}))dx}{\int_\mathbf{a}^\mathbf{b}\exp{\left(-\frac{1}{2}(\mathbf{x}-\mathbf{a})^T\Sigma_0^{-1}(\mathbf{x}-\mathbf{a})\right)}dx},
\end{align*}
where $\Sigma^{-1}=\frac{1}{\sigma_{m,0}^2}I$. Now we use Theorem~\ref{thm:mean_value_thm}. There exists a $\bm{\delta}=(\delta_1,\dots,\delta_m)^T\in\mathbb{R}^m$ such that for each $i\in\{1,2,\dots,m\}$ we have $\delta_i\in[a_i,b_i]$ and get
\begin{align*}
    \Delta C_m(\sigma_{m,0},\mathbf{c}^*) &= \frac{\exp\left(-\frac{1}{2}(\mathbf{c}^*)^T\Sigma_0^{-1}\mathbf{c}^*\right)\cdot \exp((\mathbf{c}^*)^T\Sigma_0^{-1}(\mathbf{\bm{\delta}-a}))\int_\mathbf{a}^\mathbf{b}\exp{\left(-\frac{1}{2}(\mathbf{x}-\mathbf{a})^T\Sigma_0^{-1}(\mathbf{x}-\mathbf{a})\right)}dx}{\int_\mathbf{a}^\mathbf{b}\exp{\left(-\frac{1}{2}(\mathbf{x}-\mathbf{a})^T\Sigma_0^{-1}(\mathbf{x}-\mathbf{a})\right)}dx}\\
    &= \exp\left(-\frac{1}{2}(\mathbf{c}^*)^T\Sigma_0^{-1}\mathbf{c}^*\right)\cdot \exp((\mathbf{c}^*)^T\Sigma_0^{-1}(\mathbf{\bm{\delta}-a}))\\
    &=\exp\left(-\frac{1}{2\sigma_{m,0}^2}\left((\mathbf{c}^*)^T\mathbf{c}^*+(\mathbf{c}^*)^T(\bm{\delta}-\mathbf{a})\right)\right),
\end{align*}
where the last equation holds since $\Sigma^{-1}=\frac{1}{\sigma_{m,0}^2}I$. We now plug in $\sigma_{m,0}^2=\frac{\left[||\mathbf{b-a}||_2+\frac{\Delta Q}{2}\right]\Delta Q}{\epsilon}$, and  we have

\begin{align*}
    \Delta C_m(\sigma_{m,0},\mathbf{c}^*) &= \exp\left(\frac{2(\mathbf{c}^*)^T(\bm{\delta}-\mathbf{a})-(\mathbf{c}^*)^T\mathbf{c}^*}{2\left[||\mathbf{b-a}||_2+\frac{\Delta Q}{2}\right]\Delta Q}\epsilon\right)\\
    &= \exp\left(\frac{(\mathbf{c}^*)^T(2(\bm{\delta}-\mathbf{a})-\mathbf{c}^*)}{2\left[||\mathbf{b-a}||_2+\frac{\Delta Q}{2}\right]\Delta Q}\epsilon\right). 
\end{align*}

Using the the Cauchy–Schwarz inequality, we have
\begin{align*}
    \exp\left(\frac{(\mathbf{c}^*)^T(2(\bm{\delta}-\mathbf{a})-\mathbf{c}^*)}{2\left[||\mathbf{b-a}||_2+\frac{\Delta Q}{2}\right]\Delta Q}\epsilon\right)\leq \exp\left(\frac{||\mathbf{c}^*||_2||2(\bm{\delta}-\mathbf{a})-\mathbf{c}^*||_2}{2\left[||\mathbf{b-a}||_2+\frac{\Delta Q}{2}\right]\Delta Q}\epsilon\right).
\end{align*}

With $\delta_i\in(a_i,b_i)$, $||\mathbf{c}^*||_2\leq\Delta Q$, and the triangle inequality we have
\begin{align*}
    \exp\left(\frac{||\mathbf{c}^*||_2||2(\bm{\delta}-\mathbf{a})-\mathbf{c}^*||_2}{2\left[||\mathbf{b-a}||_2+\frac{\Delta Q}{2}\right]\Delta Q}\epsilon\right)&\leq \exp\left(\frac{\Delta Q(||2(\bm{\delta}-\mathbf{a})||_2+||\mathbf{c}^*||_2)}{2\left[||\mathbf{b-a}||_2+\frac{\Delta Q}{2}\right]\Delta Q}\epsilon\right)\\
    &\leq \exp\left(\frac{(||2(\mathbf{b}-\mathbf{a})||_2+\Delta Q)}{2\left[||\mathbf{b-a}||_2+\frac{\Delta Q}{2}\right]}\epsilon\right)\\
    &=\exp(\epsilon).
\end{align*}

\section{Proof of Lemma~\ref{lem:delta_c_greater_0_nd}}\label{apdx:proof_delta_c_greater_0_nd}
Since $\Sigma_0=\sigma_{m,0}^2I$, we can simplify $\Delta C_m(\sigma_{m,0},\mathbf{c}^*)$:
\begin{align}
    \Delta C_m(\sigma_{m,0},\mathbf{c}^*) &= \frac{\int_\mathbf{a}^\mathbf{b}\exp{\left(-\frac{1}{2}(\mathbf{x}-\mathbf{a}-\mathbf{c}^*)^T\Sigma_0^{-1}(\mathbf{x}-\mathbf{a}-\mathbf{c}^*)\right)}d\mathbf{x}}{\int_\mathbf{a}^\mathbf{b}\exp{\left(-\frac{1}{2}(\mathbf{x}-\mathbf{a})^T\Sigma_0^{-1}(\mathbf{x}-\mathbf{a})\right)}d\mathbf{x}}\nonumber\\
    &= \frac{\int_{a_1}^{b_1}\cdots \int_{a_m}^{b_m}\exp\left(-\frac{\sum_{i=1}^m(x_i-a_i-c_i^*)^2}{2\sigma_{m,0}^2}\right)dx_m\cdots dx_1}{\int_{a_1}^{b_1}\cdots \int_{a_m}^{b_m}\exp\left(-\frac{\sum_{i=1}^m(x_i-a_i)^2}{2\sigma_{m,0}^2}\right)dx_m\cdots dx_1}\nonumber\\
    &= \frac{\prod_{i=1}^m\int_{a_i}^{b_i}\exp\left(-\frac{(x_i-a_i-c_i^*)^2}{2\sigma_{m,0}^2}\right)dx_i}{\prod_{i=1}^m\int_{a_i}^{b_i}\exp\left(-\frac{(x_i-a_i)^2}{2\sigma_{m,0}^2}\right)dx_i}\nonumber\\
    &= \prod_{i=1}^m\frac{\int_{a_i}^{b_i}\exp\left(-\frac{(x_i-a_i-c_i^*)^2}{2\sigma_{m,0}^2}\right)dx_i}{\int_{a_i}^{b_i}\exp\left(-\frac{(x_i-a_i)^2}{2\sigma_{m,0}^2}\right)dx_i},\label{eq:delta_c_n_simplified}
\end{align}
where $\mathbf{c}=(c_1^*,\dots,c_m^*)^T$. By Appendix~\ref{apdx:proof_delta_c_greater_0}, we have for each $i\in\{1,2,\dots,m\}$,
\begin{equation*}
    \frac{\int_{a_i}^{b_i}\exp\left(-\frac{(x_i-a_i-c_i^*)^2}{2\sigma_{m,0}^2}\right)}{\int_{a_i}^{b_i}\exp\left(-\frac{(x_i-a_i)^2}{2\sigma_{m,0}^2}\right)}>1. 
\end{equation*}

Therefore $\Delta C_m(\sigma_{m,0},\mathbf{c}^*)>1$ and $\ln(\Delta C_m(\sigma_{m,0},\mathbf{c}^*))>0$.

\section{Proof of Lemma~\ref{lem:lemma_f_decreasing_nd}}\label{apdx:proof_lemma_f_decreasing_nd}
We take the derivative of $f_m$ at~$\sigma_m$ and have
\begin{align*}
    f_m'(\sigma_m)&=2\sigma_m - \left(\frac{\left[||\mathbf{b-a}||_2+\frac{\Delta Q}{2}\right]\Delta Q}{\epsilon-\ln(\Delta C_m(\sigma_m,\mathbf{c}^*))}\right)'\\
    &= 2\sigma_m - \frac{\left[||\mathbf{b-a}||_2+\frac{\Delta Q}{2}\right]\Delta Q}{(\epsilon-\ln(\Delta C_m(\sigma_m,\mathbf{c}^*))^2}\cdot \frac{1}{\Delta C_m(\sigma_m,\mathbf{c}^*)} \cdot (\Delta C_m(\sigma_m,\mathbf{c}^*))'\\
    &= 2\sigma_m + h_m(\sigma_m),
\end{align*}

where
\begin{align*}
    h_m(\sigma_m) = - \frac{\left[||\mathbf{b-a}||_2+\frac{\Delta Q}{2}\right]\Delta Q}{(\epsilon-\ln(\Delta C_m(\sigma_m,\mathbf{c}^*))^2}\cdot \frac{1}{\Delta C_m(\sigma_m,\mathbf{c}^*)} \cdot (\Delta C_m(\sigma_m,\mathbf{c}^*))'.
\end{align*}

Since $\Delta C_m(\sigma_m,\mathbf{c}^*)>0$ for all $\sigma_m\in(0,\infty]$, then if $(\Delta C_m(\sigma_m,\mathbf{c}^*))'<0$ we will have $h_m(\sigma_m)>0$. From Equation~\eqref{eq:delta_c_n_simplified} we have
\begin{align*}
    \Delta C_m(\sigma_m,\mathbf{c}^*) = \prod_{i=1}^m\frac{\int_{a_i}^{b_i}\exp\left(-\frac{(x_i-a_i-c_i^*)^2}{2\sigma_m^2}\right)}{\int_{a_i}^{b_i}\exp\left(-\frac{(x_i-a_i)^2}{2\sigma_m^2}\right)}.
\end{align*}
Then we have the derivative:
\begin{align*}
    \frac{\partial (\Delta C_m(\sigma_m,\mathbf{c}^*))}{\partial \sigma_m}=\sum_{i=1}^m\left[\left(\frac{\partial}{\partial \sigma_m}\frac{\int_{a_i}^{b_i}\exp\left(-\frac{(x_i-a_i-c_i^*)^2}{2\sigma_m^2}\right)}{\int_{a_i}^{b_i}\exp\left(-\frac{(x_i-a_i)^2}{2\sigma_m^2}\right)}\right)\cdot \prod_{j=1,j\neq i}^m\frac{\int_{a_j}^{b_j}\exp\left(-\frac{(x_j-a_j-c_j^*)^2}{2\sigma_m^2}\right)}{\int_{a_j}^{b_j}\exp\left(-\frac{(x_j-a_j)^2}{2\sigma_m^2}\right)}\right].
\end{align*}
From Appendix~\ref{apdx:proof_lemma_f_decreasing} we have, for each $i\in\{1,2,\dots,m\}$, 
\begin{align*}
    \frac{\partial}{\partial \sigma_m}\frac{\int_{a_i}^{b_i}\exp\left(-\frac{(x_i-a_i-c_i)^2}{2\sigma_m^2}\right)}{\int_{a_i}^{b_i}\exp\left(-\frac{(x_i-a_i)^2}{2\sigma_m^2}\right)} < 0.
\end{align*}
Therefore $\frac{\partial (\Delta C_m(\sigma_m))}{\partial \sigma_m}<0$ and $f_m'(\sigma_m)>0$.

\end{document}